\documentclass[11pt]{article}

\usepackage{graphics,graphicx,amssymb,amsmath,amsthm,latexsym,xspace,cases}
\usepackage{cleveref}[2012/02/15]
\usepackage{booktabs}
\usepackage{mathtools}
\usepackage[usenames] {color}
\usepackage[mathscr] {eucal}
\usepackage{todonotes}
\usepackage{latexsym}
\usepackage{amsfonts}

\newcommand{\lab}{\operatorname{\ell}}
\newcommand{\concats}{\cdots}

\def\newdelta{\tilde\delta}
\newtheorem{theorem}{Theorem}[section]
\newtheorem{lemma}[theorem]{Lemma}
\newtheorem{claim}[theorem]{Claim}
\newtheorem{proposition}[theorem]{Proposition}

\theoremstyle{definition}

\newcommand{\ignore}[1]{}

\newcommand{\Oh}{O}

\newcommand{\ceil}[1]{\left\lceil{#1}\right\rceil}
\newcommand{\floor}[1]{\left\lfloor{#1}\right\rfloor}

\newcommand{\dfs}{\operatorname{dfs}}

\newcommand{\dist}{\operatorname{dist}}

\newcommand{\set}[1]{\left\{{#1}\right\}}

\newcommand{\range}[2]{[#1,#2]}

\newcommand{\pare}[1]{\left({#1}\right)}

\def\sF{\mathscr{F}}
\def\sG{\mathscr{G}}
\def\sC{\mathscr{C}}

\def\newdelta{\tilde\delta}
\def\mask{\mathrm{mask}}
\def\decode{\mathrm{\textsc{Decode}}}

\def\SumIntegers{\mathrm{SumIntegers}}
\def\root{\mathrm{root}}

\def\PreFix{\mathrm{PreFix}}

\def\MicroSum{\mathrm{MicroSum}}
\def\MacroSum{\mathrm{MacroSum}}
\def\MicroRoot{\mathrm{MicroRoot}}
\def\sT{\mathscr{T}}
\def\parent{\mathrm{parent}}

\def\code{\mathrm{code}}

\def\beginsmall#1{\vspace{-\parskip}\begin{#1}\itemsep-\parskip}
\def\endsmall#1{\end{#1}\vspace{-\parskip}}

\title{Simpler, faster and shorter labels for distances in graphs}

\author{Stephen Alstrup   \thanks{Depart. of Computer Science,
University of Copenhagen, Denmark. E-mail:
{\tt stephen.alstrup.private@gmail.com}.}
\and Cyril Gavoille \thanks{LaBRI - Universit\'e de Bordeaux, France. E-mail:
{\tt gavoille@labri.fr}.} \and
Esben Bistrup Halvorsen \thanks{Depart. of Computer Science,
University of Copenhagen, Denmark. E-mail:
{\tt esben@bistruphalvorsen.dk}.} 
\and  Holger Petersen \thanks{E-mail: {\tt dr.holger.petersen@googlemail.com}.}}

\begin{document}

\maketitle

\vspace*{-15pt}

\begin{abstract}
\noindent%
We consider how to assign \emph{labels} to any undirected graph with $n$ nodes such that, given the labels of two nodes and no other information regarding the graph, it is possible to determine the distance between the two nodes.  The challenge in such a \emph{distance labeling scheme} is primarily to minimize the maximum label lenght and secondarily to minimize the time needed to answer distance queries (decoding). Previous schemes have offered different trade-offs between label lengths and query time. This paper presents a simple algorithm with shorter labels and shorter query time than any previous solution, thereby improving the state-of-the-art with respect to both label length and query time in one single algorithm. Our solution addresses several open problems concerning label length and decoding time and is the first improvement of label length for more than three decades.

More specifically, we present a distance labeling scheme with labels of length $\frac{\log 3}{2}n+o(n)$ bits\footnote{Throughout the paper, all logarithms are in base 2.} and constant decoding time. This outperforms all existing results with respect to both size and decoding time, including Winkler's (Combinatorica 1983) decade-old result, which uses labels of size $(\log 3)n$ and $\Oh(n / \log n)$ decoding time, and Gavoille et al.\  (SODA'01), which uses labels of size $11n+o(n)$ and $O(\log\log n)$ decoding time.  In addition, our algorithm is simpler than the previous ones.
In the case of integral edge weights of size at most $W$, we present almost matching upper and lower bounds for the label size $\ell$:
$\frac{1}{2}(n-1) \log \ceil{\frac{W}{2} + 1} \le \ell \le \frac{1}{2}n \log{(2W+1)} + O(\log n\cdot\log(nW))$.
Furthermore, for $r$-additive approximation labeling schemes, where distances can be off by up to an additive constant $r$,  we  present both upper and lower bounds. In particular, we present an upper bound for $1$-additive approximation schemes which, in the unweighted case, has the same size (ignoring second order terms) as an adjacency labeling scheme, namely $n/2$. 
We also give results for bipartite graphs as well as for exact and $1$-additive distance oracles.
\end{abstract}

\newpage

\section{Introduction}

A \emph{distance labeling scheme} for a given family of graphs assigns \emph{labels} to the nodes of each graph from the family such that, given the labels of two nodes in the graph and no other information, it is possible to determine the shortest distance between the two nodes. The labels are assumed to be composed of bits. The main goal is to make the worst-case label size as small as possible while, as a subgoal, keeping query (decoding) time under control. 
The problem of finding implicit representations with small labels for specific families of graphs was first introduced by Breuer~\cite{Breuer66,BF67}, and efficient labeling schemes were introduced in~\cite{KNR92,muller}.

\subsection{Distance labeling}

For an undirected, unweighted graph,  a na\"ive solution to the distance labeling problem is to let each label be a table with the $n-1$ distances to all the other nodes, giving labels of size around $n\log n$ bits. For graphs with bounded degree $\Delta$ it was shown~\cite{BF67} in the 1960s that labels of size $2n\Delta$ can be constructed such that two nodes are adjacent whenever the Hamming distance~\cite{hamming} of their labels is at most $4\Delta-4$. In the 1970s, Graham and Pollak~\cite{grahampollak} proposed to label each node with symbols from $\{0,1,*\}$, essentially representing nodes as corners in a ``squashed cube'', such that the distance between two nodes exactly equals the Hamming distance of their labels (the distance between $*$ and any other symbol is set to 0). They conjectured the smallest dimension of such a squashed cube (the so-called \emph{Squashed cube conjecture}), and their conjecture was subsequently proven by Winkler~\cite{winkler} in the 1980s. This reduced the label size to $\ceil{(n-1) \log 3}$, but the solution requires $\Oh(n / \log n)$ query time to decode distances. Combining~\cite{KNR92} and~\cite{moon1965minimal} gives a lower bound of $\ceil{n/2}$ bits. A different distance labeling scheme of size of $11n+o(n)$ and  with $O(\log \log n)$ decoding time was proposed in~\cite{Gavoille200485}. The article also raised it as open problem to find the right label size. Later in~\cite{weinmannpeleg} the algorithm from~\cite{Gavoille200485} was modified, so that the decoding time was further reduced to $O(\log^* n)$ with slightly larger labels, although still of size $O(n)$. This article raised it as an open problem whether the query time can be reduced to constant time. Having distance labeling with short labels and simultaneous fast decoding time is a problem also addressed in text books such as~\cite{spinrad2003efficient}. Some of our are solutions are simple enough to replace material in text books.

Addressing the aforementioned  open problems, we present a distance labeling scheme with labels of size $\frac{\log 3}{2}n+o(n)$ bits and with constant decoding time. See \Cref{tab:unweightedgraph} and \Cref{picture} for an overview.

\begin{table}[h!]
\centering
\makebox[0pt][c]{
\begin{tabular}{|c|c|c|c|}
\hline
\bf  Space & \bf Decoding time & \bf Year & \bf Reference  \\
\hline
 $(\log 3)n$ & $\Oh(n / \log n)$ & 1972/1983 & ~\cite{grahampollak,winkler} \\
\hline
$11n$ & $\Oh(\log \log n)$ & 2001 & ~\cite{Gavoille200485}\\
\hline
$cn, c>11$ & $O(\log^* n)$ & 2011 & ~\cite{weinmannpeleg} \\
\hline
$\frac{\log 3}{2}n$ & $\Oh(1)$ & 2015 & this paper \\
\hline
\end{tabular}
}
\caption{Unweighted undirected graphs. Space is listed presented without second order terms. A graphical presentation of the results is given in \Cref{picture}}
\label{tab:unweightedgraph}
\end{table}

\begin{figure}
\begin{tikzpicture}[x=2.3cm, y=0.9cm] 
\draw   (4, 4);
\draw [fill] (1, 1) circle (2pt);
\node at (1.5,1) {This paper};
\draw [fill] (4, 2) circle (2pt);
\node at (4.5,2) {1972/1983};
\draw [fill] (3, 3) circle (2pt);
\node at (3.3,3) {2001};
\draw [fill] (2, 4) circle (2pt);
\node at (2.3,4) {2011};
\draw [thick,->] (0, 0) -- (4.5, 0);
\draw [thick,->] (0, 0) -- (0, 4.5);
\node at (5, 0) {Time};
\node at (1,-0.3) {$\Oh(1)$};
\node at (2,-0.3) {$\Oh(\log^*n)$};
\node at (3,-0.3) {$\Oh(\log \log n)$};
\node at (4,-0.3) {$\Oh(n / \log n)$};

\node at (-0.4,1) {$\frac{\log 3}{2}n$};
\node at (-0.35,2) {$(\log 3)n$};
\node at (-0.25,3) {$11n$};
\node at (-0.3,4) {$\Oh(n)$};
\node at (0, 5) {Space};
\end{tikzpicture}
\caption{A graphical representation of the results from \Cref{tab:unweightedgraph}.} \label{picture}
\end{figure}
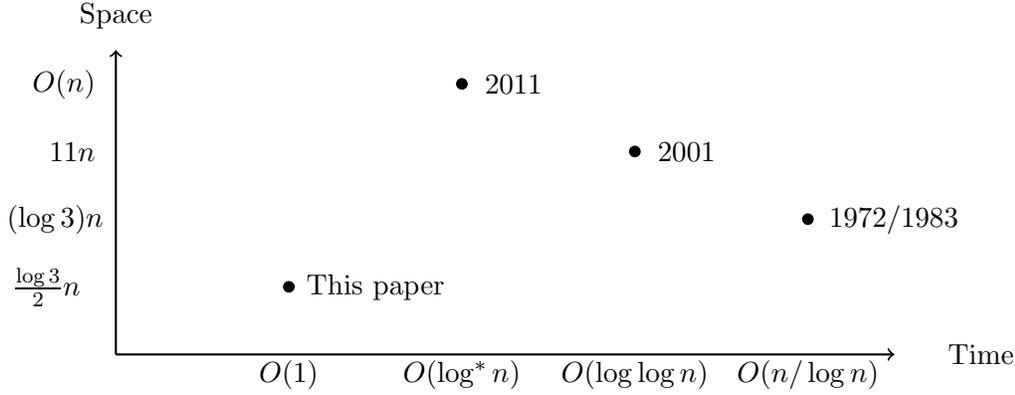

Distance labeling schemes for various families of graphs exist, e.g., for
trees~\cite{alstrupbillerauhe,peleg}, bounded tree-width~\cite{Gavoille200485}, distance-hereditary~\cite{GP03b}, bounded clique-width~\cite{CV03}, some non-positively
curved plane~\cite{CDV06}, interval~\cite{GP08} and
permutation graphs~\cite{BG09}. In~\cite{Gavoille200485} it is proved that distance labels require $\Theta(\log^2 n)$ bits for trees, $\Oh(\sqrt{n}\log n)$ and $\Omega(n^{1/3})$ bits for planar graphs, and $\Omega(\sqrt{n})$ bits for bounded degree graphs. In an unweighted graph, two nodes are adjacent iff their distance is $1$. Hence, lower bounds for adjacency labeling apply to distance labeling as well, and adjacency lower bounds can be achieved by reduction~\cite{KNR92} to induced-universal graphs, e.g. giving $\frac{n}{2}$ and $\frac{n}{4}$ for general and bipartite graphs, respectively. An overview of  adjacency labeling can be found in~\cite{AlstrupKTZ14}.

Various computability requirements are sometimes imposed on  labeling schemes~\cite{AKM01,KNR92,siamcompKatzKKP04}. This paper assumes the RAM model and mentions the time needed for decoding in addition to the label size. 

\subsection{Overview of  results}
For weighted graphs we assume integral edge weights from $[1,W]$. Letting each node save the distance to all other nodes would require a scheme with labels of size $\Oh(n\log (nW))$ bits. Let $\dist_G(x,y)$ denote the shortest distance in $G$ between nodes $x$ and $y$. An $r$-additive approximation scheme returns a value $\dist'_G(x,y)$, where $\dist_G(x,y)\leq \dist'_G(x,y) \leq \dist_G(x,y)+r$.

Throughout this paper we will assume that $\log W=o(\log n)$ since otherwise the  na\"ive solution mentioned above will be as good as our solution. Ignoring second order terms, we can for general weighted graphs and
constant decoding time achieve upper and lower bounds for label length
as stated in \Cref{tab:overview}.
For bipartite graphs we also show a lower bound of $\frac{1}{4}n\log\floor{2W/3 + 5/3}$ and an upper bound of $\frac{1}{2}n$ whenever $W=1$.

\begin{table}[h!] 
\centering
\makebox[0pt][c]{
\begin{tabular}{|c|c|c|}
\hline
\bf Problem & \bf Lower bound  & \bf Upper bound \\
\hline
General graphs  & $\frac{1}{2}(n-1)\log\ceil{W/2+1}$  & $\frac{1}{2}n\log (2W+1)$ \\
\hline
\end{tabular}
}
\caption{General graphs with weights from $[1,W]$, where $\log W=o(\log n)$. The upper bound has an extra $o(n)$ term, and decoding takes constant time.
}
\label{tab:overview}
\end{table}

We present, as stated in \Cref{table:secondterms}, several trade-offs between decoding time, edge weight $W$, and space needed for the second order term.  

\begin{table}[h!] 
\centering
\makebox[0pt][c]{
\begin{tabular}{|c|c|c|}
\hline
\bf  Time & \bf Second order term & \bf W \\
\hline
N/A & $\Oh(\log n \cdot \log (nW))$ & Any value \\
\hline
 $\Oh(n)$ & $\Oh(\log^2 n)$ &  $\Oh(1)$ \\
\hline
$\Oh(1)$ & $\Oh(\frac{n}{\log n}\log (2W+1)(\log \log n +\log W))$ & $2^{o(\log n)}$\\
\hline
\end{tabular}
}
\caption{Second order term for the upper bound for general graphs (in \Cref{tab:overview}). 
The results also hold for the $n/2$ labels in the unweighted, bipartite case.
 It may be possible to relax the restriction $W=\Oh(1)$ if the word "finite" in \Cref{access} below from~\cite{DPT10} does not mean ``constant''.}
\label{table:secondterms}
\end{table}

We also show that, for any $k,D\geq 0$ with $\log k=o(\log n)$ and $D\leq 2(k+1)W-1$, there exists a $(2kW+\ceil{\frac{D}{2(k+1)W-D}})$-additive distance scheme using labels of size $\frac{1}{2(k+1)}n \log (2(k+1)W+1-D) + \Oh(\log n \cdot \log (nW))$ bits. 

Finally, we present lower bounds for approximation schemes. In particular, for $r < 2W$ we prove that
labels of $\Omega(n\log{(W/(r+1))})$ bits are required for an $r$-additive distance labeling scheme.

\subsection{Approximate distance labeling schemes and oracles}
Approximate distance labeling schemes are well studied; see
e.g.,~\cite{Gavoille200485,GuptaKL03,gupta2005traveling,peleg,ThZw05}. For instance, graphs of doubling dimension~\cite{Talwar04} and planar graphs~\cite{Thorup2004distance} both enjoy schemes with polylogarithmic label length which return approximate distances below a $1+\varepsilon$ factor of the exact distance. Approximate schemes that return a small additive error have also been
investigated, e.g. in~\cite{CDEHV08,GL05,KL06}. In~\cite{GKKPP01}, lower and upper bounds for $r$-additive schemes, $r\leq 2$, are given for chordal, AT, permutation and interval graphs. 
 For general graphs the current best lower bound~\cite{GKKPP01} for $r \ge 2$-additive scheme is  $\Omega(\sqrt{n/r})$. For $r=1$, one needs $\frac{1}{4}n$ bits since a $1$-additive scheme can answer adjacency queries in bipartite graphs. Using our approximative result, we achieve, by setting $k=0$ and $D=W=1$,  a $1$-additive distance labeling scheme which, ignoring second order terms, has the same size (namely $\frac{1}{2}n$ bits) as an optimal adjacency labeling scheme. Somehow related, \cite{BCE03} studies labeling schemes that preserve exact distances between nodes with minimum distance $P$, giving an $O((n/P)\log^2{n})$ bit solution.

Approximate distance oracles introduced in~\cite{ThZw05} use a global table (not necessarily labels) from which approximate distance queries can be answered quickly. One can na\"ively use the $n$ labels in a labeling scheme as a distance oracle (but not vice versa). For unweighted graphs, we achieve constant query time for $1$-additive distance oracles using $\frac{1}{2}n^2+o(n^2)$ bits in total, matching (ignoring second order terms) the space needed to represent a graph. Other techniques only reduce space for $r$-additive errors for $r>1$.  For exact distances in weighted graphs, our solution achieves $\frac{1}{2}n^2 \log{(2W+1)} + o(n^2)$ bits for $\log W=o(\log n)$. This relaxes the requirement of $W=\Oh(1)$ in~\cite{FerraginaNV10} (and slightly improves the space usage in that paper). 
 
 \subsection{Second order terms are important}
 Chung's solution in~\cite{Chung90} gives labels of size $\log n+O(\log \log n)$ for adjacency
 labeling in trees, which was improved to $\log n + O(\log^* n)$  in~\cite{alstruprauhe} and in~\cite{bonichon2006short,Fraigniaud2009randomized,fraigniaudkorman2,KMS02}
 to $\log n + \Oh(1)$ for various special cases.
A recent STOC'15 paper~\cite{AlstrupKTZ14} improves label size for adjacency in generel graphs from $n/2+O(\log n)$ to $n/2+ O(1)$.
Likewise, the second order term for ancestor relationship is improved in a sequence of STOC/SODA papers~\cite{AKM01,AR02,AlstrupBR03,fraigniaudkorman2,fraigniaudkorman} (and~\cite{abiteboul})  to $\Theta(\log \log n)$, giving labels of size $\log n+\Theta(\log \log n)$.

Somewhat related, \emph{succinct data structures} (see, e.g.,~\cite{DPT10,FarzanM13,FarzanM14,MunroRRR12,patrascu08succinct})
 focus on the space used in addition to the information theoretic lower bound, which is often a lower order term with respect to the overall space used.

\subsection{Labeling schemes in various settings and applications}
By using labeling schemes, it is possible to avoid costly access to large global tables, computing instead locally and distributed. Such properties are used, e.g., in XML search engines~\cite{AKM01}, network routing and distributed algorithms~\cite{Cowen01,EilamGP03,throupzwick,ThZw05}, dynamic and parallel settings ~\cite{CohenKaplan2010,dynamicKormanP07}, graph representations~\cite{KNR92}, and other applications~\cite{siamcompKatzKKP04,Korman2010,peleg2,peleg,SK85}. From the SIGMOD, we see labeling schemes used in~\cite{AIY13,JRXL12} for shortest path queries and in~\cite{CHWWx13} for reachability queries.  Finally, we observe that compact $2$-hop labeling (a specific distance labeling scheme) is central for computing exact distances on real-world networks with millions of arcs in real-time~\cite{DGSW14}. 

\subsection{Outline of the paper} \Cref{sec:warmup} illustrates some of our basic techniques. \Cref{seclong,secshort} present our upper bounds for exact distance labeling schemes for general graphs. \Cref{secapprox} presents upper bounds for approximate distances. Our lower bounds are rather simple counting arguments with reduction to adjacency and have been placed in \Cref{SecLower}.

\section{Preliminaries}

\paragraph{Trees.}

Given a rooted tree $T$ and a node $u$ of $T$, denote by $T_u$ be the
subtree of $T$ consisting of all the descendants of $u$ (including
itself). The \emph{depth} of $u$ is the number of edges on the unique simple path from $u$
to the root of $T$. For any rooted subtree $A$ of $T$, denote by $\root(A)$
the root of $A$, as the node of $A$ with smallest depth. Denote by $A^* = A \setminus \set{\root(A)}$ the forest obtained from
$A$ by removing its root. Denote by $|A|$ the number of nodes of
$A$: hence, $|A^*|$ represents its number of edges.  Denote by
$\parent_T(u)$ the parent of the node $u$ in $T$. Let 
$T[u,v]$ denote the nodes on the simple path from $u$ to $v$ in $T$. The variants $T(u,v]$ and $T[u,v)$
denote the same path without the first and last node, respectively.

\paragraph{Graphs.}
Throughout we assume graphs to be connected. If a graph is not connected, we can add $O(\log{n})$ bits to each
label, indicating the connected component of the node, and then handle
components separately. We denote by $\dist_G(u,v)$ the minimum distance
(counted with edge weights) of a path in $G$ connecting the nodes $u$ and $v$.

\paragraph{Representing numbers and accessing them.}
We will need to encode numbers with base different from~$2$ and sometimes compute prefix sums on a sequence of numbers. We apply some existing results:

\begin{lemma} [\cite{Makinen2007}]\label{Oprefix}
  A table with $n$ integral entries in $[-k,k]$ can be represented in
  a data structure of $\Oh(n \log k)$ bits to support prefix sums in
  constant time.
\end{lemma}

\begin{lemma}[\cite{DPT10}]\label{access}
  A table with $n$ elements from a finite alphabet $\sigma$ can be represented
  in a data structure of $\ceil{n \log |\sigma|}$ bits, such that any element of the table can
  be read or written in constant time. The data structure requires
  $\Oh(\log n)$ precomputed word constants.
\end{lemma}

\begin{lemma} [simple arithmetic coding]\label{deltaer}
  A table with $n$ elements from an alphabet $\sigma$ can be represented
  in a data structure of $\ceil{n \log |\sigma|}$ bits.
\end{lemma}

\section{Warm-up} \label{sec:warmup}
This section presents, as a warm-up, a distance labeling scheme which does not achieve the strongest combination of label size and decoding time, but which uses some of the techniques that we will employ later to achieve our results.
For nodes $x,u,v$, define
\[
\delta_x(u,v) = \dist_G(x,v) - \dist_G(x,u).
\]
Note that the triangle inequailty entails that
\[
-\dist_G(u,v)\leq \delta_x(u,v)\leq \dist_G(u,v).
\]
In particular, $\delta_x(u,v)\in [-W,W]$ whenever $u,v$ are adjacent.

Given a a path $v_0,\dots ,v_t$ of nodes in $G$, the telescoping property of $\delta_x$-values means that
\[
\delta_x(v_0,v_t) = \sum_{i=1}^{t}\delta_x(v_{i-1},v_i).
\]
Since $v_{i-1}$ and $v_i$ are adjacent, we can encode the $\delta_x$-values above as a table with $t$ entries, in which each entry is a an element from the alphabet $[-W,W]$ with $2W+1$ values. Using Lemma~\ref{deltaer} we can encode this table with $\ceil{t\log(2W+1)}$ bits. 
Note that we can 
compute $\dist_G(x,v_t)$ from $\dist_G(x,v_0)$ by adding a prefix sum of the sequence of $\delta_x$-values: 
\[
\dist_G(x,v_t)=\dist_G(x,v_0)+\sum_{i=1}^t\delta_x(v_{i-1},v_i).
\]

The \emph{Hamiltonian number} of $G$ is the number $h(G)$ of edges of a Hamiltonian walk in $G$, i.e.\ a closed walk of minimal length (counted without weights) that visits every node in $G$. It is well-known that $n\leq h(G)\leq 2n-2$, the first inequality being an equality iff $G$ is Hamiltonian, and the latter being an equality iff $G$ is a tree (in which case the Hamiltonian walk is an Euler tour); see~\cite{MR2082480,goodman-hedetniemi1973}.

Consider a Hamiltonian walk $v_0,\dots ,v_{h-1}$ of length $h=h(G)$. Given nodes $x,y$ from $G$, we can find $i,j$ such that $x=v_i$ and $y=v_j$. Without loss of generality we can assume that $i\leq j$. If $j \leq i+h/2$, we can compute $\dist_G(x,y)$ as the sum of at most $\floor{h/2}$ $\delta_x$-values:
\[
\dist_G(x,y) = \dist_G(v_i,v_j)=\sum_{k=i}^{j-1}\delta_{x}(v_k,v_{k+1}).
\]
If, on the other hand, $j> i+h/2$, then we can compute $\dist_G(x,y)$ as the sum of at most $\floor{h/2}$ $\delta_y$-values:
\[
\dist_G(x,y) = \dist_G(v_j,v_i)  
= \sum_{k=j}^{i-1}\delta_{y}(v_k,v_{k+1}),
\]
where we have counted indices modulo $h$ in the last expression.
This leads to the following distance labeling scheme. For each node $x$ in $G$, assign a label $\lab(x)$ consisting of
\begin{itemize}
\item a number $i \in [0,h-1]$ such that $x=v_i$; and
\item the $\floor{h/2}$ values $\delta_x(v_k,v_{k+1})$ for $k=i,\dots ,i+\floor{h/2}-1\pmod {h}$.
\end{itemize}
From the above discussion it follows that the labels $\lab(x)$ and $\lab(y)$ for any two nodes $x,y$ are sufficient to compute $\dist_G(x,y)$.

We can encode $\lab(x)$ with $\ceil{\frac{1}{2}h\log (2W+1)} +\ceil{\log h}$ bits using Lemma~\ref{deltaer}. If $G$ is Hamiltonian, this immediately gives a labeling scheme of size $\ceil{\frac{1}{2}n\log(2W+1)}+\ceil{\log n}$. In the general case, we get size
$\ceil{(n-1)\log(2W+1)}+\ceil{\log n}$, which for $W=1$ matches Winkler's~\cite{winkler} result when disregarding second order terms.
\Cref{theo:esben} in the next section shows that it is possible to obtain labels of size $\frac{1}{2}n\log(2W+1)+O(\log n \cdot \log (nW))$ even in the general case. \Cref{constant} in the section that follows shows that we can obtain constant time decoding with $o(n)$ extra space.

\section{A scheme of size $\frac{1}{2}n\log(2W+1)$} \label{seclong}

We now show how to construct a distance labeling scheme of size $\frac{1}{2}n\log(2W+1)+\Oh(\log n\cdot\log(nW))$. 

First, we recall the \emph{heavy-light decomposition} of trees~\cite{sleatortarjan}. 
Let $T$ be a rooted tree. The nodes of $T$ are classified as either \emph{heavy} or \emph{light} as follows. The root $r$ of $T$ is light. For each non-leaf node $v$, pick one child  $w$ where $|T_w|$ is maximal among the children of $v$ and classify it as heavy; classify the other children of $v$ as light. 
The \emph{apex} of a node $v$ is the nearest light ancestor of $v$.
By removing the edges between light nodes and their parents, $T$ is divided into a collection of \emph{heavy paths}. Any given node $v$ has at most $\log n$ light ancestors (see~\cite{sleatortarjan}), so the path from the root to $v$ goes through at most $\log n$ heavy paths.

Now, enumerate the nodes in $T$ in a depth-first manner where heavy children are visited first. Denote the number of a node $v$ by $\dfs(v)$. Note that nodes on a heavy path will have numbers in consecutive order; in particular, the root node $r$ will have number $\dfs(r)=0$, and the nodes on its heavy path will have numbers $0,1,\dots$. Assign to each node $v$ a label $\lab_T(v)$ consisting of the sequence of $\dfs$-values of its first and last ancestor on each heavy path, ordered from the top of the tree and down to $v$. Note that the first ancestor on a heavy path will be the apex of that heavy path and will be light, whereas the last ancestor on a heavy path will be the parent of the apex of the subsequent heavy path. This construction is similar to the one used in~\cite{alstrupnca2014} for nearest common ancestor (NCA) labeling schemes, although with larger sublabels. Indeed, the label $\lab_T(v)$ is a sequence of at most $2\log n$ numbers from $[0,n[$. We can encode this sequence with $O(\log^2 n)$ bits.

Suppose that the node $v$ has label $\lab_T(v)=(l_1,h_1,\dots ,l_t,h_t)$, where $l_1=\dfs(r)=0$ and $h_t=\dfs(v)$ and where $l_i,h_i$ are the numbers of the first and last ancestor, respectively, on the $i$'th heavy path visited on the path from the root to $v$. Since nodes on heavy paths are consecutively enumerated, it follows that the nodes on the path from the root to $v$ are enumerated
\[
0=l_1,\dots ,h_1,l_2,\dots ,h_2,\dots ,l_t,\dots,h_t,
\]
where duplicates may occur in the cases where $l_i=h_i$, which happens when the first and last ancestor on a heavy path coincide.

In addition to the label $\lab_T(v)$, we also store the label $\lab_T'(v)$ consisting of the sequence of distances $\dist_G(l_i,v)$ and $\dist_G(h_i,v)$. This label is a sequence of at most $2\log n$ numbers smaller than $nW$, and hence we can encode $\lab'_T(v)$ with $\Oh(\log n\cdot \log(nW))$ bits. Combined, $\lab_T(v)$ and $\lab'_T(v)$ can be encoded with $\Oh(\log n\cdot \log(nW))$ bits. 

Now consider a connected graph $G$ with shortest-path tree $T$ rooted at some node $r$. Using the above enumeration of nodes, we can construct a distance labeling scheme in the same manner as in  \Cref{sec:warmup}, except that instead of using a Hamiltonian path, we use the $\dfs$-enumeration of nodes in $T$ from above, and we save only $\delta_x$-value between nodes and their parents, using $\ceil{\frac{1}{2}n\log(2W+1)}$ bits due to \Cref{deltaer}. More specifically, for each node $x$, we assign a label $\lab(x)$ consisting of
\begin{itemize}
\item the labels $\lab_T(x)$ and $\lab'_T(x)$ as described above; and
\item the $\floor{n/2}$ values $\delta_x(\parent(v),v)$ for all $v$ with $\dfs(x) < \dfs(v)\leq \dfs(x)+\floor{n/2} \pmod n$.
\end{itemize}
We can encode the above with $\frac{1}{2}n\log(2W+1)+O(\log n\cdot \log(nW))$ bits.

Given nodes $x\neq y$, either $\lab(x)$ will contain $\delta_x(\parent(y),y)$ or $\lab(y)$ will contain $\delta_y(\parent(x),x)$. Without loss of generality, we may assume that $\lab(x)$ contains $\delta_x(\parent(y),y)$. Let $z$ denote the nearest common ancestor of $x$ and $y$. Note that $z$ must be the last ancestor of either $x$ or $y$ on some heavy path, meaning that $\dfs(z)$ appears in either $\lab_T(x)$ or $\lab_T(y)$. By construction of depth-first-search, a node $v$ on the path from (but not including) $z$ to (and including) $y$ will have a $\dfs$-number $\dfs(v)$ that satisfies the requirements to be stored in $\lab(x)$. Thus, $\lab(x)$ must, in fact, contain $\delta_x$-values for all nodes in $T_{(z,y]}$.

Next, note that, since $T$ is a shortest-path tree, $\dist_G(x,z)= \dist_T(x,z)$. Now, if $z$ appears in $\lab_T(x)$, we can obtain $\dist_T(x,z)$ directly from $\lab'_T(x)$; else, $z$ must appear in $\lab_T(y)$, and we can then obtain $\dist_T(z,y)$ from $\lab_T'(y)$ and compute $\dist_T(x,z)=\dist_T(x,r)-\dist_T(r,y)+\dist_T(z,y)$. In either case, we can now compute the distance in $G$ between $x$ and $y$ as
\[
\dist_G(x,y) = \dist_G(x,z) + \sum_{v\in T(z,y]} \delta_x(\parent(v),v).
\]
The label of $x$ contains all the needed $\delta_x$-values, and $\lab_T(x)$  and $\lab_T(y)$ combined allows us to determine the $\dfs$-numbers of the nodes on $T(z,y]$, so that we know exactly which $\delta_x$-values from $x$'s label to pick out. Thus we have proved:
\begin{theorem} \label{theo:esben}
There exists a distance labeling scheme for graphs with label size $\frac{1}{2}n\log (2W+1) +O(\log n\cdot \log(nW))$.
\end{theorem}

This gives us the first row of~\Cref{table:secondterms}. To obtain the second row, we encode the $\delta_x$ values with \Cref{access}. Doing this we can access each value in constant time and simply traverse in $\Oh(n)$ time the path from $y$ to $z$, adding $\delta_x$ values along the way. Note, however, that \Cref{access} only applies for $W=\Oh(1)$. 
Saving the $\delta_x$-values in a prefix sum structure as described in \Cref{Oprefix}, we can  compute the sum using $\log n$ look-ups. The next section describes how we can avoid spending $O(\log n)$ time (or more) on this, while still keeping the same label size.

For unweighted ($W=1$), bipartite graphs, $\delta_x$-values between adjacent nodes can never be $0$, which means that we only need to consider two rather than three possible values. Thus, we get label size $\frac{1}{2}n +O(\log^2 n)$ instead in this case. 
We shall give no further mention to this in the following.

\section{Constant query time}\label{secshort}
Let $T$ be any rooted spanning tree of the connected graph $G$ with $n$ nodes. We create an edge-partition $\sT = \set{T_1, T_2, \dots}$ of $T$ into rooted subtrees, called \emph{micro trees}. Each micro tree has at most $\beta$ edges, and the number of micro trees is $|\sT| = O(n/\beta)$. We later choose the value of $\beta$. For completeness we give a proof (Lemma~\ref{lempart}) in the appendix of the existence of such a construction. Observe that the collection $\set{T_i^*}_{i\ge 1}$ forms a partition
of the nodes of $T^*$. As the parent relationship in $T_i$ coincides with the one of $T$, we have  $\parent_{T_i}(u) = \parent_T(u)$ for all $u \in T^*_i$.

For every node $u\in T^*$, we denote by $i(u)$ the unique index $i$ such that $u \in T^*_i$. For a node $u$ of $T^*$ we let $\MicroRoot(u)=\root(T_{i(u)})$, and for $r=\root(T)$ let $\MicroRoot(r)=r$.

Define the \emph{macro tree} $M$ to have node set $\{\MicroRoot(u)\mid u\in G\}$ and an edge between $\MicroRoot(u)$ and $\MicroRoot(\MicroRoot(u))$ for all $u\neq r$.

By construction, $M$ has $O(n/\beta)$ nodes.

Our labeling scheme will compute the distance from $x$ to $y$ as
\[
\dist_G(x,y)=\dist_G(x,r) + \delta_x(r,\MicroRoot(y))+\delta_x(\MicroRoot(y),y).
\]
The first addend, $\dist_G(x,r)$, is saved as part of $x$'s label using $\log n+\log W$ bits. The second addend can be computed as a sum of $\delta_x$-values for nodes in the macro tree and is hence referred to as the \emph{macro sum}. The third addend can be computed as a sum of $\delta_x$-values for nodes inside $y$'s micro tree and is hence referred to as the \emph{micro sum}. The next two sections explain how to create data structures that allow us to compute these values in constant time.

\subsection{Macro sum}
Consider the macro tree $M$ with $\Oh(n/\beta)$ nodes. As mentioned in Section~\ref{sec:warmup} there exists a Hamiltonian walk $v_0,\dots ,v_{h-1}$ of length $h=\Oh(n/\beta)$, where we can assume that $v_0=r$.
Given nodes $x,y\in G$, consider a path in $M$ along such a Hamiltonian walk from $r$ to $\MicroRoot(y)$. This is a subpath $v_0,\dots , v_t$ of the Hamiltonian walk, where $t$ is chosen such that $v_t=\MicroRoot(y)$. Note that 
\[
\delta_x(r,\MicroRoot(y)) = \delta_x(v_0,v_t) = \sum_{i=0}^{t-1} \delta_x(v_i,v_{i+1}).
\]
Since each edge in $M$ connects two nodes that belong to the same micro tree, and the distance within each micro tree is a most $\beta W$, we have that $\delta_x(v_i,v_{i+1})\in [-\beta W,\beta W]$ for all $i$. Using Lemma~\ref{Oprefix} we can store these $\delta_x$-values in a data structure, $\PreFix_x$, of size $\Oh((n/\beta)\log(2\beta W+1)) = \Oh(n\log(\beta W)/\beta)$ such that prefix sums can be computed in constant time. This data structure is stored in $x$'s label. An index $t$ with $v_t=\MicroRoot(y)$ is stored in $y$'s label using $\Oh(\log(n/\beta))$ bits. These two pieces of information combined allow allows us to compute $\delta_x(r,\MicroRoot(y))$ for all $y$.

\emph{Label summary:} 
For a (pre-selected) Hamiltonian walk $v_0,\dots ,v_{h-1}$ in $M$, we store in the label of each node $x$ a datastructure $\PreFix_x$ of size $\Oh(n\log(\beta W)/\beta)$ such that prefix sums in the form $\sum_{i=0}^{t-1} \delta_x(v_i,v_{i+1})$ can be computed in constant time.
In addition, we store in the label of $x$ an index $m(x)$ such that $v_{m(x)}=\MicroRoot(x)$, which requires $\Oh(\log(n/\beta))$ bits. 
 
\subsection{Micro sum}
For any node $v\neq r$, define
\[
\delta_x(v) = \delta_x(\parent_T(v),v)
\]
Note that, for a node $y\in T^*_i$, $\delta_x(\MicroRoot(y),y)$ is the sum of the values $\delta_x(v_j)$ for all nodes $v_j\in T_i^*$ lying on the path from $\MicroRoot(y)$ to $y$. Each of these $\delta_x$-values is a number in $[-W ,W]$.

For each $i$,  order the nodes in $T_i^*$ in any order. For each node $x$ and index $i$, let $\delta_x(T^*_i)  = (\delta_x(v_1),\dots ,\delta_x(v_{|T_i^*|}))$, where $v_1,\dots ,v_{|T^*_i|}$ is the ordered sequence of nodes from $T_i^*$. We will construct our labels such that $x$'s label stores $\delta_x(T^*_i)$ for half of the total set of delta values (we will see how in the next section), and such that $y$'s label stores information about for which $j$'s the node $v_j$ lies on the path between $\MicroRoot(y)$ and $y$. With these two pieces of information, we can compute $\delta_x(\MicroRoot(y),y)$ as described above.

We define $f(W)=2W+1$. The sequence $\delta_x(T_i^*)$ consists of $|T_i^*|$ values from $[-W,W]$ can be encoded with $|T_i^*|\ceil{\log f(W)}$ bits. To store this more compactly, we will use an injective function, as described in Lemma~\ref{deltaer} that maps every sequence of $t$ integers from $[-W,W]$ into a bit string of length $\ceil{t\log{f(W)}}$. Denote by $\code(\delta_x(T^*_i))$  such an encoding of the sequence $\delta_x(T^*_i)$ to a bit string oflength $\ceil{|T_i^*|\log f(W)}\leq \ceil{\beta \log f(W)} $, as $|T_i^*|\leq \beta $

In order to decode the encoded version of $\delta_x(T^*_i)$ in constant time, we construct a tabulated inverse function $\code^{-1}$. 
From the input and output sizes, we see that we need a table with $2^{\ceil{\beta \log{(f(W))}}}$ entries, for each of the $\beta$ possible micro tree sizes, and each result entry having $\beta \ceil {\log f(W)}$ bits, giving a total space of $\beta 2^{\ceil{\beta \log{f(W)}}} \beta \ceil {\log f(W)}$ bits. 

Let $T_i=T_{i(y)}$. Let $\&$ be the bitwise AND operator. In node $y$'s label we save the bit string $\mask(y)$ such that $\mask(y)$ $\&$ $\delta_x(T^*_i)$ gives an integer sequence $S$ identical to $\delta_x(T^*_i)$, except that the integer $\delta_x(v)$ has been replaced by $0$ for all $v$ that are not an ancestor of $y$. Given $S$ we can now compute the micro sum $\delta_x(\MicroRoot(y),y)$ as the sum of integers in the sequence $S$. We will create a tabulated function that sums these integers, $\SumIntegers$. 
$\SumIntegers$ is given a sequence of up to $\beta$ values in $[-W,W]$, and the output is a number in $[-\beta W,\beta W]$. We can thus tabulate $\SumIntegers$ as a table with $\beta 2^{\beta \ceil{\log f(W)}}$ entries each of size $\ceil{\log f(\beta W)}$, giving a total space of $\beta 2^{\beta \ceil{\log f(W)}}\ceil{\log f(\beta W)}$.

Both functions, $\code^{-1}$ and $\SumIntegers$, have been tabulated in the above. A lookup in a tabulated function can be done in constant time on the RAM as long as both input and output can be represented by $\Oh(\log n)$ bits. We can achieve this by setting
\[
\beta \leq \frac{c \log n}{\ceil{\log f(W)}}
\]
for a constant $c$.
To see this, note that the maximum of the four input and output values above is $\ceil{\log \beta}+\beta \ceil{\log f(W)}$. Using the above inequality then gives $\log \log n+c \log n=\Oh(\log n)$.

The tables for the tabulated functions are the same for all nodes. Hence, in principle, assuming an upper bound for $n$ is known, we could encode the two tables in global memory, not using space in the labels. However, as we will see, the tables take no more space than the prefix table $\PreFix_x$, so we can just as well encode them into the labels. Doing that we use an additional $\beta 2^{\ceil{\beta \log{(f(W))}}} \beta \ceil {\log f(W)}$ for the $\code^{-1}$ table and $\beta 2^{\beta \ceil{\log f(W)}} \ceil{\log f(W\beta)}$ for the $\SumIntegers$ table. Using that $W =o(n)$ and substituting $\beta$ for the above expression then gives, after a few reductions, that the extra space used is no more than $O((\log n)^4 n^c)$ bits. Since the prefix table uses at least $\Oh(\frac{n \log \log n}{\log n})$ bits, we see that the added space does not (asymptotically) change the total space usage, as long as we choose $c<1$.

\emph{Label summary:} We will construct the labels such that either $x$'s label contains $\delta_x(T^*_i(y))$ or vice versa (we shall see how in the next section). Using the tabulated function $\code^{-1}$, the bits in $\delta_x(T^*_i(y))$ can be extracted in constant time from $x$'s label. Using $\mask(y)$ from $y$'s label and the tabulated function $\SumIntegers$, we can then compute $\delta_x(\MicroRoot(y),y)$ in constant time. The total space used for all this is no more than $\Oh(\frac{n \log \log n}{\log n})$.

\subsection{Storing and extracting the deltas}
Let the micro trees in $\sT$ bee given in a specific order: $T_1,\dots , T_{|\sT|}$.
Let $D(x) = \code(\delta_x(T^*_1)) \concats \code(\delta_x(T^*_{|\sT|}))$ denote the binary string composed of the concatenation of each
string $\code(\delta_x(T^*_i))$ in the order $i=1,2,\dots,|\sT|$. 

Let $L = |D(x)|$ be the length in bits of $D(x)$. Let $p_i \in
[0,L)$ be the position in the string $D(x)$ where the substring
$\code(\delta_x(T^*_i))$ starts. E.g., $D(x)[0] = D(x)[p_1]$ is the
first bit of $\code(\delta_x(T^*_1))$, $D(x)[p_2]$ the first bit of
$\code(\delta_x(T^*_2))$, and so on. 
According to \Cref{deltaer} we have $p_i = \sum_{j<i}
\ceil{|T^*_j|\log{f(W)}}$. Observe that the position
$p_i$ only depends on $i$ and $W$ and not on $x$.

We denote by $a(y)$ and $a'(y)$ the starting and ending positions
of the substring $\code(\delta_x(T^*_{i(y)}))$ in
$D(x)$. More precisely, $a(y) = p_{i(y)}$ and $a'(y) = p_{i(y)+1}-1$, so
that $|\code(\delta_x(T^*_{i(y)}))| = a'(y) - a(y)+1$. For each node $y$ we use $\Oh(\log n)$ bits to store  $a(y)$ and $a'(y)$ in its label. 

For a node $x$ we will only save approximate half of $D(x)$, in a table $H(x)$. $H(x)$ will start with $\code(\delta_x(T^*_{i(x)}))$ and the code for the following micro trees in the given circular order until $H(x)$ in total has at least $n/ 2$ $\delta_x$ values, but as few a possible.  In other words $H(x) =\code(\delta_x(T^*_{i(x)}))  \concats \code(\delta_x(T^*_{j(x)}))$ where the indexes $i(x),
i(x)+1, \dots, j(x)$ may wrap to~$1$ after reaching the largest index
$|\sT|$ if $j(x) < i(x)$. Let $b(x) = p_{j(x)+1}$. 

In a node $x$'s label we save $a(x), a'(x), b(x)$ and  $L$ using $\Oh(\log n)$ bits. Having those values we know which $\delta_x$ values from $D(x)$ are saved in $x$'s label as well as the position of them in $H(x)$. Furthermore we know the position of the $\delta_x$-values of $x$'s own micro tree in $D(x)$.
 We will need to extract at most $\ceil{\beta \log f(W)}=\Oh(\log n)$ consecutive bits from $H(x)$ in one query. On the word-RAM this can be done in constant time.

\begin{proposition}\label{prop:H}
  Let $x,y$ be two nodes of $G$. Then,
  \beginsmall{itemize}
\item[\rm(i)] $|H(x)|=
  \frac{1}{2}n\log{f(W)} + O(\frac{n}{\log n}\log f(W))$; and 
\item[\rm(ii)] either $\code(\delta_x(T^*_{i(y)}))$ is part of $H(x)$ or $\code(\delta_y(T^*_{i(x)}))$ is part of $H(y)$.
  \endsmall{itemize}
\end{proposition}

\begin{proof}
  Let $\sT '$ be the subset of $\sT$ encoded in  $H(x)$. We have:
\begin{eqnarray*}
  |H(x)| &=& \sum_{{T_i} \in \sT '} \ceil{|T^*_i| \log{f(W)}} ~<~ 
  \sum_{{T_i} \in \sT '} \pare{|T^*_i| \log{f(W)} + 1} \\
  &<& \frac{1}{2}n \log{f(W)} + |\sT| +\ceil{\beta \log f(W)}~<~ \frac{1}{2}n\log{f(W)} + O(n/\beta + \beta\log{f(W)}) \\
  &<& \frac{1}{2}n\log{f(W)} +  O(\frac{n}{\log n}\log f(W))
\end{eqnarray*}
which proves (i). Part (ii) follows from the fact that $x$ saves at least half of the $\delta_x$'s in a cyclic order. If $y$ not is include here, $x$ must be included in the $\delta_y$-values saved by $y$.
\end{proof}

\subsection{Summary}
The label of $x$ is composed of the follows items.

\begin{enumerate}

\item The values $a(x)$, $a'(x)$, $\mask(x)$, $m(x)$, $\dist_G(x,r)$, $L$ and $b(x)$: $O(\log{n})$.

\item A prefix table, $\PreFix_x$, for the values in the macro tree: $O(\frac{n}{\log n}((\log f(W))^2+\log \log n \log f(W)))$. 

\item The table $H(x)$: $\frac{1}{2}n\log{f(W)} +  O(\frac{n}{\log n}\log f(W))$.

\item Global tables, $\code^{-1}$ and $\SumIntegers$ of size $O(\frac{n
    \log \log{n}}{\log n})$.
\end{enumerate}

Note that $L$ and the global tables are common to all the nodes. In addition we may need to use $\Oh(\log n)$ bits to save the start position in the label for the above constant number of sublabels. 

\begin{lemma}\label{lem:length}
  Every label has length at most $\frac{1}{2}n\log{f(W)} +
  O(\frac{n}{\log n}(\log^2 W+\log \log n \log f(W)))$ bits.
\end{lemma}

Let $\decode(\ell(x,G),\ell(y,G))$ denote the distance returned by the
decoder given the labels of $x$ and of $y$ in $G$. It is defined by:

\medskip

\begin{center} 
\fbox{
\begin{minipage}{.92\textwidth}
\vspace{1ex}
$\decode(\ell(x,G),\ell(y,G))$: 
\beginsmall{enumerate}
\item If $(a(x)\leq a(y)<b(x)) \vee (b(x)<a(x)\leq a(y)) \vee (a(y)<b(x)<a(x))$ then $s=a(y)-a(x) \pmod L$ and $e=a'(y)-a(x) \pmod L$
\item Else return $\decode(\ell(y,G),\ell(x,G))$
\item $\MacroSum=\PreFix_x(m(y))$
\item $S=\code^{-1}(H_x[s, \dots ,e])$ $\&$ $\mask(y)$
\item $\MicroSum=\SumIntegers(S)$
\item Return $\dist_G(x,r) + \MicroSum + \MacroSum$
\endsmall{enumerate}
\end{minipage}
}
\end{center}

\begin{theorem}\label{constant}
  There exists a distance labeling scheme for graphs with edge weights in $[1,W]$ using labels of length $\frac{1}{2}n\log{(2W+1)} +
  O(\frac{n}{\log n}\log (2W+1)(\log W+\log \log n ))$ bits and constant decoding time.
\end{theorem}

\section{Approximate distances}  \label{secapprox}

By considering only a subset of nodes from $G$ and using the previous techniques, it is possible to create an approximation scheme where the label size is determined by a smaller number of nodes but with larger weights between adjacent nodes. We leave the details for \Cref{nodeapprox} and present here only the result.

\begin{theorem} 
There exists a $(2kW)$-additive distance labeling scheme for graphs with $n$ nodes and edge weights in $[1,W]$ using labels of size $\frac{1}{2(k+1)}n\log (2(k+1)W+1) +O(\log n\cdot \log(nW))$.
\end{theorem}

Another way to achieve an approximation scheme is to use a smaller set of weights while keeping the accumulated error under control. This leads to the following result whose proof can be seen in \Cref{numberapprox}.

\begin{theorem} \label{approxdelta}
For any $D\leq 2W-1$ there exists a $\ceil{\frac{D}{2W-D}}$-additive distance labeling scheme for graphs with $n$ nodes and edge weights in $[1,W]$ using labels of size $\frac{1}{2}n \log (2W+1-D) +O(\log n\cdot \log(nW))$.
\end{theorem}

One instance of~\Cref{approxdelta} is $D=W$, which gives a 1-additive distance labeling scheme of size $\frac{1}{2}n \log (W+1) +o(n)$. For $D=2W-1$ we get a $(2W-1)$-additive distance labeling scheme of size $\frac{1}{2}n+o(n)$. For constant $r$ the above technique also applies to our constant time decoding results. For unweighted graphs this implies that we can have labels of size $\frac{1}{2}n+o(n)$ with a 1-additive error and constant decoding time. 

By combining the above two theorems, we obtain the theorem below; see \Cref{finalapprox}.

\begin{theorem} 
For any $k\geq 0$ and $D\leq 2(k+1)W-1$ there exists a $(2kW+\ceil{\frac{D}{2(k+1)W-D}})$-additive distance labeling scheme for graphs with $n$ nodes and edge weights in $[1,W]$ using labels of size $\frac{1}{2(k+1)}n \log (2(k+1)W+1-D) +O(\log n\cdot \log(nW))$ bits.
\end{theorem}

\newpage
\bibliographystyle{plain}
\bibliography{LowDistLabels}

\makeatletter
\def\runninghead{\hrulefill\quad APPENDIX\quad\hrulefill}
\def\ps@headings{
\def\@oddhead{\footnotesize\rm\hfill\runninghead\hfill}}
\def\@evenhead{\@oddhead}
\def\@oddfoot{\rm\hfill\thepage\hfill}\def\@evenfoot{\@oddfoot}
\makeatother

\newpage
\setlength{\headsep}{15pt} \pagestyle{headings}

\appendix

\section{Lower bounds} \label{SecLower}

Our lower bound technique can be seen as a generalization of the
classical counting argument for adjacency labeling schemes. Indeed, for
$r=0$ and $W=1$, our formula yields $(n-1)/2$ bits, which is exactly the number of bits needed for
adjacency. 
The lower bound we develop here is well-suited for
small additive errors $r$. In particular, when $r < 2W$ we prove that
labels of $\Omega(n\log{(W/(r+1))})$ bits are required for an $r$-additive distance labeling scheme.

Given an unweighted graph $B$ and an integer $W \ge 1$, denote by $\sF_W(B)$ be
the family of all subgraphs of $B$ whose edges are weighted by values
taken from $\range{1}{W}$.

\begin{theorem}\label{th:lower}
  Let $B$ be an unweighted graph with $n$ vertices, $m$ edges and girth at least $g$,
  and let $r,W$ be integers such that $r \in [0,(g-2)W)$. Then, every
  $r$-additive approximate distance labeling scheme for $\sF_W(B)$
  requires a total label length of at least $m \log{(k+1)}$, and thus
  a label of at least $(m/n) \log{(k+1)}$ bits, where
  \[
  k ~=~ \floor{\frac{g-2}{g-1} \cdot \pare{\frac{W}{r+1} + 1}}~.
  \]
\end{theorem}

\begin{proof}
  An $r$-approximate distance matrix for a weighted graph $G$ with
  vertex-set $\range{1}{n}$ is an $n \times n$ matrix $M$ such that
  $\dist_G(x,y) \le M[x,y] \le \dist_G(x,y) + r$ for all vertices
  $x,y$ of $G$. 

  The basic idea of our lower bound technique is to show that $\sF_W(B)$
  contains a large set $\sG$ of weighted graphs for which no two
  graphs can have the same $r$-approximate distance matrix. A crude
  observation is that an $r$-approximate distance matrix for each
  graph of $\sG$ can be generated from the ordered list of all the
  labels provided by any $r$-additive approximate labeling scheme for
  $\sG$. So, it turns out that, by a simple counting argument, the
  total label length must be at least $\log{|\sG|}$. In particular,
  the labeling scheme must assign, for some vertex of some graph of
  $\sG$, a label of at least $(\log{|\sG|})/n$ bits. We now construct
  such a set $\sG$ with $|\sG| = (k+1)^m$.

  For the shake of the presentation, define $W_i = W -
  (k-i-1)(r+1)$ for $i=0,\dots ,k$. Note that the $W_i$s increase 
  with $i$, and more precisely that $W_{i+1} = W_i + r+1$. Moreover,
  we observe that:

  \begin{claim}\label{claim:W}
    The following hold: $k\ge 1$, $W_0, \dots, W_{k-1} \in [1,W]$, and $W_k \le (g-1)W_0$.
  \end{claim}

  Before we give a formal proof of Claim~\ref{claim:W} (which is a
  basic calculation), we explain how to derive our lower bound.

  Consider the set $\sC$ of all edge-colorings of $B$ into $k+1$
  colors. More precisely, an edge-coloring $c \in\sC$ is simply a
  function $c: E(B) \to \range{0}{k}$ mapping to each edge $e$ of $B$
  some integer $c(e) \in\range{0}{k}$. Clearly, $|\sC| = (k+1)^m$
  since each of the $m$ edges of $B$ can receive $k+1$ distinct
  values.

  With each coloring $c \in \sC$, we associate a weighted graph $G$
  with edge-weight function $w$ obtained from graph $B$ by testing the
  color of each edge $xy$ of $B$. If $c(xy) = k$, the edge is deleted.
  And, if $c(xy) = i <k$, we keep $xy$ in the graph and set $w(xy) =
  W_i$. We denote by $\sG$ the family of graphs constructed by this
  process from all the colorings of $\sC$. It is clear that, given
  $B$, one can recover from $G$ and $w$ the initial coloring $c$ (just
  scan all the possible edges of $B$, check if they exist in $G$
  and look at their weights). In other words the construction is
  bijective and thus $|\sG| = |\sC| = (k+1)^m$.

  By construction, each graph of $\sG$ is a subgraph of $B$. Moreover,
  by Claim~\ref{claim:W}, each weight is some integer $W_i \in [1,W]$
  as $i \in\range{0}{k-1}$ (edges of color $k$ have been removed). In
  other words, $\sG \subseteq \sF_W(B)$. It remains to prove that any
  two graphs of $\sG$ cannot have the same $r$-approximate distance
  matrix. The intuition is that two graphs of $\sG$ differ only when 
  there is an edge $xy$ in $B$ whose color is different in the two 
  graphs. Because of the choice of the edge-weights, the distance
  between $x$ and $y$ in the two graphs must, as we shall see, differ by at least $r+1$.

  Let $G,G'$ be two distinct weighted graphs of $\sG$. Denote by
  $w,w'$ their respective edge-weighting functions, and by $M,M'$ any
  $r$-approximate distance matrices for $G$ and $G'$ respectively. As
  we will show, if $G$ and $G'$ are different, there must exist an
  edge $xy$ of $B$, and two colors $i,j \in\range{0}{k}$, $i<j$, such
  that $\dist_G(x,y) \le W_i$ and $\dist_{G'}(x,y) \ge W_j$ (the case
  $\dist_{G'}(x,y) \le W_i$ and $\dist_{G}(x,y) \ge W_j$ is symmetric).
  For this purpose, we consider two cases:

  (i) The graphs $G,G'$ are different because there is an edge $xy$ of $B$
  which  is in $G$ but not in $G'$. We have $xy \in E(G)$, which implies that
  $\dist_{G}(x,y) \le w(xy) \le W_{k-1}$, since the color of $xy$ in
  $G$ is $< k$. Further, $xy \notin E(G')$ implies that $\dist_{G'}(x,y) \ge
  (g-1) W_0$, since any path from $x$ to $y$ in $G'$ contains at least
  $g-1$ edges ($G'$ is a subgraph of $B$ which has girth at least $g$),
  and the minimum weight assigned to any edge is $W_0$. 
(Note that this holds, in particular, when $x$ and $y$ are unconnected and $\dist_G(x,y)=\infty$.)
Thus from
  Claim~\ref{claim:W}, $\dist_{G'}(x,y) \ge W_k$.
 So the claim holds for $i=k-1$ and $j=k$.

  (ii) The graphs $G,G'$ are different because there is an edge $xy$ in $G$
  and $G'$ with different weights. Assuming that $w(xy) < w'(xy)$, there must exist 
 $i,j$ such that $w(xy) = W_i$ and $w'(xy) = W_j$. Note that $i <
  j < k$. We have $\dist_{G}(x,y) = W_i$ and $\dist_{G'}(x,y) = W_j$
  since we have seen in the previous case that every path from $x$ to
  $y$ and excluding the edge $xy$ has cost at least $(g-1)W_0$. By
  Claim~\ref{claim:W}, $W_i < W_j < W_k \le (g-1)W_0$.

  In both cases we have found $i,j \in \range{0}{k}$, $i<j$, such that
  $\dist_G(x,y) \le W_i$ and $\dist_{G'}(x,y) \ge W_j$. Now, by
  definition of $M$ and $M'$, $M[x,y] \le \dist_G(x,y) +r \le W_i + r$
  and $W_j \le \dist_{G'}(x,y) \le M'[x,y]$. Since $W_j \ge W_{i+1} =
  W_i + r+1$, we conclude that $M[x,y] < M'[x,y]$ proving that no
  two different graphs of $\sG$ can have the same $r$-approximate
  distance matrix.

  To complete the proof, it remains to prove Claim~\ref{claim:W}.  Let
  us first show that $k\ge 1$ (which is required since in the proof we
  use for instance that $W_0 \le W_{k-1}$). Recall that $r \le (g-2)W
  -1$,

  \[
  k ~= \floor{ \frac{g-2}{g-1} \pare{ \frac{W}{r+1} + 1}} ~\ge~
  \floor{ \frac{g-2}{g-1} \pare{ \frac{W}{(g-2)W} + 1} } ~=~ \floor{
    \frac{g-2}{g-1} \pare{ \frac{1}{g-2} + 1}} ~=~ 1.
  \]

  Let us show that $W_0 \ge 1$. Since $W_0 = W - (k-1)(r+1)$, it
  suffices to check that $(k-1)(r+1) < W$. We have,

  \[
  (k-1)(r+1) ~=~ \pare{ \floor{ \frac{g-2}{g-1} \pare{ \frac{W}{r+1} +
        1}} - 1} \cdot \pare{r+1} ~\le~ \frac{g-2}{g-1} \cdot \frac{W}{r+1}
  \cdot (r+1) ~<~ W
  \]
  since the girth $g$ is always at least three.

  Now we have $W_0,\dots,W_{k-1} \in [1,W]$, since the $W_i$'s are
  non-decreasing, $W_0 \ge 1$, and $W_{k-1} = W - (k-(k-1)-1)(r+1) = W$.

  Let us show that $W_k \le (g-1)W_0$. We have $W_k = W + r+1$ and
  $W_0 = W - (k-1)(r+1)$. Therefore,
  \begin{eqnarray*}
    W_k \le (g-1)W_0 &\Leftrightarrow& W + r+1 \le (g-1)(W-(k-1)(r+1)) \\
    &\Leftrightarrow& r+1 + (g-1)(k-1)(r+1) \le (g-2)W \\
    &\Leftrightarrow& (g-1)(k-1) \le (g-2) \frac{W}{r+1} - 1\\
    &\Leftrightarrow& k \le \frac{g-2}{g-1} \cdot \frac{W}{r+1} -
    \frac{1}{g-1} + 1 ~=~ \frac{g-2}{g-1} \cdot\pare{\frac{W}{r+1} +1}~.
  \end{eqnarray*}
  The latter equation is true by the choice of $k$. This completes the
  proof of Claim~\ref{claim:W} and of Theorem~\ref{th:lower}.
\end{proof}

\begin{table}[htb] 
  \begin{center}
        \renewcommand{\arraystretch}{1.75}
        \begin{tabular}{|c||c|c|c|c|}
          \hline
          Graphs & $r=0, W\ge 1$ & $r=1, W\ge 2$ & $r=0, W=1$ & $r=(g-2)W-1$\\
          \hline\hline
          General &
          $\frac{1}{2}(n-1)\log{\ceil{\frac{W}{2} + 1}}$ &
          $\frac{1}{2}(n-1)\log{\floor{\frac{W}{4} + \frac{3}{2}}}$ &
          \multicolumn{2}{c|}{$\frac{1}{2}(n-1)$} \\
          \hline
          Bipartite &
          $\frac{1}{4}n\log{\floor{\frac{2W}{3} + \frac{5}{3}}}$ &
          $\frac{1}{4}n\log{\floor{\frac{W}{3} + \frac{5}{3}}}$ &
          \multicolumn{2}{c|}{$\frac{1}{4}n$} \\
          \hline
        \end{tabular}  
        \caption{
	  Lower bounds derived from Theorem~\ref{th:lower}. For
          ``general graphs'' we use the family $\sF_W(K_n)$, where
          $K_n$ denotes the complete graph on $n$ vertices, so $m =
          n(n-1)/2$ and $g=3$. For ``Bipartite graphs'' we use the
          family $\sF_W(K_{n/2,n/2})$, where $K_{n/2,n/2}$ denotes the
          complete bipartite graph on $n$ vertices (assuming $n$ even)
          so $m = n^2/4$ and $g=4$. Note that the case $r=W=1$ and the
          case $r=2W-1$ is captured by the last column of the last
          line, and so the lower bound is $n/4$.}
    \label{tab:lower}
  \end{center}
\end{table}

A collection of corollaries to \Cref{th:lower} can be seen in \Cref{tab:lower}.

The case $r\ge 2$ and $W=1$ is out of the range of our lower bound, as
long as we choose for $B$ a graph with $m = \Theta(n^2)$ edges. Our lower
bound still applies for $r=2,3$ and $W=1$, but using girth-$6$ graphs
$B$ that are known to exists with $m = \Theta(n^{3/2})$ edges. 
There are several constructions, based on finite projective geometries, of graphs with $\Omega(n^{3/2})$ edges and girth at least 6 (see for instance~\cite{Wenger91}). So,
Theorem~\ref{th:lower} can also prove the $\Omega(\sqrt{n}\,)$ lower
bound for $r=2,3$ and $W=1$. The case of larger $r$ can be captured by
the more general lower bound of~\cite{GKKPP01}, that uses a
subdivision technique, and shows that $\Omega(\sqrt{n/(r+1)}\,)$ bit
labels are required for any $r \ge 2$.

\section{Constructing micro trees}
\begin{lemma}\label{lempart}
  Let $k$ be a positive integer. Every $m$-edge tree has an
  edge partition into at most $\ceil{m/k}$ trees of at most $2k$
  edges.
\end{lemma}

\begin{proof}
 Consider a tree $T$ with $m$ edges. If $T$ has fewer than $k$ edges,
 then the partition is $T$ itself and we are done.

 Otherwise, we will construct a subtree $A$ of $T$ with at least $k$
 and at most $2k$ edges such that $T-A$ is still
 connected. (By $T-A$ we mean the forest induced by all the
 edges in $E(T)- E(A)$.) Once such an $A$ is constructed, we
 can repeat the process on the remaining tree $T-A$ until
 having a tree with less than $k$ edges. Since each subtree $A$ has
 at least $k$ edges, the process stop after we have constructed at most
 $\ceil{m/k}$ trees.
\end{proof}

\section{Approximate distances}  \label{approx}

\subsection{Approximation using fewer nodes}\label{nodeapprox}

\begin{lemma} \label{kapprox}
Given a graph $G$ with $n$ nodes, edge weights in $[1,W]$ and a rooted spanning tree $T$, we can, for integers $k\geq 0$, construct a tree $T(k)$ whose node set is a subset of $T$ and with the following properties.
\begin{itemize}
\item $|T(k)| \leq 1+ \frac{n}{k+1}$.
\item For any node $v \in T(k)^*$,  $\dist_G(v,parent_{T(k)}(v)) \leq (k+1)W$.
\item For any node $w \in G$, there exists a node $v \in T(k)$ with $\dist_G(v,w) \leq kW$.
\end{itemize}
\end{lemma}
\begin{proof}
Partition the nodes in $T$ into $k+1$ equivalence classes according to their depth in $T$ modulo $k+1$. One of these equivalence classes must contain $\lfloor n/(k+1)\rfloor$ or fewer nodes. Select such a subset of nodes and denote it $T(k)$. Also include the root of $T$ in $T(k)$, giving that $|T(k)|\leq 1+n/(k+1)$. In $T(k)$ construct an edge between two nodes iff no other nodes from $T(k)$ are on the simple path between the nodes in $T$. Then, for any $v\in T(k)^*$, $\dist_G(v,\parent_{T(k)}(v)) \leq (k+1)W$ since the number of edges between $v$ and $\parent_{T(k)}(v)$ in $T$ is at most $k+1$. Similarly, for any $w \in T$, its nearest ancestor $v$, which also in $T(k)$ (could be $w$ itself) is at most $k$ edges up in $T$, giving $\dist_G(v,w) \leq kW$
\end{proof}

Roughly speaking, the approximation scheme presented here applies the previous techniques to create an exact distance labeling scheme for $T(k)$, which has fewer nodes but larger weights between adjacent nodes. For a node $x$ that is not in $T(k)$, we find its nearest ancestor $x'$ in $T(k)$ and give $x$ the same label as $x'$. We then approximate $\dist_G(x,y)$ by $2kW+\dist_G(x',y')$. This will at most give us an error in  $[0,4kW]$, meaning that we now have a $(4kW)$-additive labeling scheme with labels of size $\frac{n}{2(k+1)}\log (2(k+1)W+1)$, ignoring second order terms.  It is possible to optimize this approach and obtain a $(2kW)$-additive scheme, by only using approximate distance for either $x$ or $y$ to their nearest ancestors $x'$ or $y'$. 
Here, we show how to do it for the heavy path approach. A similar result holds for the micro tree approach.

As described in \Cref{seclong}, the label $\lab(x)$ includes the sublabels $\lab_T(x)$ and $\lab'_T(x)$ using $O(\log n\cdot \log(nW))$ bits. Our new label for approximate distance will include those sublabels as well. For a node $v$ in $T$ let $v'$ be its nearest ancestor in $T(k)$. In $x$'s label we also save $\lab_T(x')$ and $\lab'_T(x')$. In addition, we include a label $\lab_T''(x)$ containing the distances from $x$ to $w'$ for all $w$ that appear in $\lab(x)$; this label also uses $O(\log n\cdot \log(nW))$ bits.

Now. if $x$ or $x'$ is an ancestor to $y$ in $T$, we compute $\dist_G(x,y)$ as $kW$ plus the distance in $T$ from $x$ or $x'$ to $y$. This will at most give an additive error of $2kW$. Similarly for $y$ and $y'$ and $x$. Those computation can be done as explained in \Cref{seclong} using the labels defined so far. If we cannot compute the distance in this way, then, consider the nearest common ancestor $z$ of $x$ and $y$ in $T$. The node $z'$ must then be the nearest common ancestor of $x'$ and $y'$ in $T(k)$.   
Let $n_k=|T(k)|$. We will save  $\floor{n_k/2}$ of the values $\delta_x(\parent_{T(k)}(v),v)$ for all $v$ from $T(k)$ with $\dfs(x') < \dfs(v)\leq \dfs(x')+\floor{n_k/2} \pmod {n_k}$. As $\delta_x(\parent_{T(k)}(v),v) \in [-(k+1)W,(k+1)W]$ we can encode all these $\delta$-values using $\ceil{\frac{1}{2}n_k \log f((k+1)W)}$ bits. We can now compute $\dist_G(x,y) = kW+\dist_G(x,z') + \sum_{v\in T(k)(z',y']} \delta_x(\parent_{T(k)}(v),v)$, where $\dist_G(x,z')$ can be computed from $\lab_T''(x)$. This will at most give an additive error of $2kW$. We have now proved the following theorem.

\begin{theorem} \label{theo:esbenapprox}
There exists a $(2kW)$-additive distance labeling scheme for graphs with $n$ nodes and edge weights in $[1,W]$ with labels of size $\frac{1}{2(k+1)}n\log f((k+1)W) +O(\log n\cdot \log(nW))$.
\end{theorem}

\subsection{Approximation using fewer weights}\label{numberapprox}
With edge weights in $[1,W]$ we have been encoding $\frac{1}{2}n$ numbers from the interval $I=[-W,W]$ of size $|I|=2W+1$ using $\frac{1}{2}n\log (2W+1)$ bits. The theorem below uses an approximation technique where we use integers from a smaller set $I' \subseteq I$, which will reduce the space consumption but introduce an error when computing $\delta_x$-values. As we shall see, the error can be capped even when we are summing many $\delta_x$-values.

\begin{theorem} 
For any $D\leq 2W-1$ there exists a $\ceil{\frac{D}{2W-D}}$-additive distance labeling scheme for graphs with edge weights in $[1,W]$ using labels of size $\frac{1}{2}n \log (2W+1-D) +O(\log n\cdot \log(nW))$ bits.
\end{theorem}
\begin{proof}
Let us create a subset $I'\subseteq I$ with $|I'|=|I|-D$. In $I'$ we always include the maximum and minimum from $I$, and hence we require $D\leq |I|-2 =2W-1$. In addition we minimize the maximum number $Q$ of consecutive numbers from $I-I'$. Hence, for $i_1 \in I'$ (excluding maximum) there exists a number $i_2\in I'$ such that $i_2\leq i_1 + Q +1$. Since $|I'|=|I|-D=2W+1-D$, we have $2W-D$ pair of neighbors in $I'$, where we by ``neighbors'' mean two numbers in $I'$ with no other number from $I'$ between them. By equally spreading the $D$ missing numbers between the $2W-D$ pairs, we can obtain $Q=\ceil{\frac{D}{2W-D}}$. By substituting values in $I'$ for values in $I$, we can now encode $t$ values from $I$ with $\ceil{t \log (2W+1-D)}$ bits. This will introduce an error, but in the case of $\delta_x$-values, the accumulated error can be kept below $Q$ as described below.

Let $T$ be a tree and consider the $\delta_x$-values $\delta_x(v)=\delta_x(\parent(v),v)$ for nodes $v\in T^*$. Each $\delta_x$-value is a number in $I=[-W,W]$. We will visit the nodes top down starting from (but not including) the root $r$ and assigning to each node $y$ a new approximate value:  $\newdelta_x(y) \in I'$. (For the root $r$ we implicitly associate the value $0$. Implicitly, since it may not be a value in $I'$.)
Recall that $\delta_x(r,y) = \sum_{v \in T[y,r)} \delta_x(v)$, and define $\newdelta_x(r,y) =\sum_{v \in T[y,r)} \newdelta_x(v)$. We will assign the values such that $\delta_x(r,y) \leq \newdelta_x(r,y) \leq \delta_x(r,y)+Q$.

For $y \in T^*$, let $A(y)=\newdelta_x(r,y)-\delta_x(r,y)$. We prove by induction that $A(y) \in [0,Q]$. So assume inductively that $A(\parent(y))\in [0,Q]$. 
If $\delta_x(y) \in I'$, we can define $\newdelta_x(y)=\delta_x(y)$, and we then have $A(y)=A(\parent(y))\in [0,Q]$ as desired. 
If this is not the case, let $i_1,i_2\in I'$ be the largest and smallest numbers from $I'$, respectively, with $i_1< \delta_x(y) < i_2$. By assumption, $i_2-i_1\leq Q+1$.
If $A(\parent(y))+i_1-\delta_x(y)\geq 0$, we can set $\newdelta_x(y)=i_1$ and obtain $A(y)\in [0,Q]$. If not, then we must have $A(\parent(y))+i_2-\delta_x(y)< i_2-i_1\leq Q+1$, so we can set $\newdelta_x(y)=i_2$ and obtain $A(y)\in [0,Q]$.
This concludes the theorem.

Above we have been changing all $\delta_x$-values top-down from the root. In the constant time solution, we could instead change the values top-down for each micro tree $T_i$, keeping exact distances to the root and in the macro tree.
\end{proof}

\subsection{Final approximation} \label{finalapprox}
We can combine the above two approaches by, for a $k\geq 0$, first using \Cref{nodeapprox} to obtain a $(2kW)$-additive distance labeling scheme with labels of size $\frac{n}{2(k+1)}\log f((k+1)W) +O(\log n\cdot \log(nW))$. The approximate scheme will use edge weights in $[1,(k+1)W]$ to which then can apply the technique from \Cref{numberapprox}, finally getting: 

\begin{theorem} \label{approxfinal}
	For any $k\geq 0$ and $D\leq 2(k+1)W-1$ there exists a $(2kW+\ceil{\frac{D}{2(k+1)W-D}})$-additive distance labeling scheme for graphs with $n$ nodes and edge weights in $[1,W]$ using labels of size $\frac{n}{2(k+1)} \log (2(k+1)W+1-D) +O(\log n\cdot \log(nW))$ bits.
\end{theorem}

\end{document}